\documentclass[12pt,a4paper]{article}

\usepackage{amsmath,amssymb,amsthm}
\usepackage{mathrsfs}
\usepackage{graphicx}
\usepackage{hyperref}
\usepackage{tikz}
\usetikzlibrary{arrows.meta,decorations.pathmorphing,decorations.markings,
                decorations.pathreplacing,positioning,shapes.geometric,calc,patterns}
\tikzset{->-/.style={decoration={markings,mark=at position #1 with {\arrow{Stealth}}},postaction={decorate}}}
\usepackage{xcolor}
\usepackage{physics}
\usepackage{slashed}
\usepackage{mathtools}
\usepackage{cite}
\usepackage{geometry}
\newcommand{\1}{\mathbf{1}}

\renewcommand{\Tr}{\mathrm{Tr}}
\newcommand{\im}{\mathrm{Im}}

\newcommand{\soft}{\mathrm{soft}}
\newcommand{\hard}{\mathrm{hard}}
\newcommand{\tot}{\mathrm{tot}}

\newcommand{\incl}{\mathrm{incl}}
\newcommand{\virt}{\mathrm{virt}}
\newcommand{\realp}{\mathrm{real}}

\theoremstyle{plain}
\newtheorem{proposition}{Proposition}
\theoremstyle{remark}

\definecolor{forwardcol}{RGB}{0,100,200}
\definecolor{backwardcol}{RGB}{200,50,0}
\definecolor{cutcol}{RGB}{0,160,60}
\definecolor{photoncol}{RGB}{150,0,200}

\title{\textbf{Soft QED as Open Quantum System:\\
Infrared Cancellation and Soft-Shell Coarse Graining}}
\author{\bf Soo-Jong Rey\\
Theory Group, Kwangwoon University, Seoul, Korea}
\date{}

\begin{document}
\maketitle

\begin{abstract}
I formulate the unresolved soft-photon sector of QED as an open quantum system. The resolved charged particles and hard photons form the system, photons below an energy resolution $\Lambda$ form the environment, and the basic object is the reduced density matrix.  A resolved outcome $f$, whose projection $\Pi_f$ may contain any finite collection of hard particles and photons, has probability
$P_{f}(i) =\sum_n|\langle f,n|S|i,0\rangle|^2=\langle i|F_f|i\rangle$, with Kraus operators $K_n={}_\soft\!\langle n|S|0\rangle_\soft$ and effect $F_f=\sum_nK_n^\dagger\Pi_fK_n$.  The Schwinger--Keldysh closed-time-path formulation places virtual and unresolved-real terms in one doubled-contour expansion. At one loop they carry the same on-shell eikonal kernel with opposite signs. This elegantly organizes the QED probability: for the same observable, perturbative order, diagrams, and phase space, the open quantum system calculation gives the same infrared-finite terms as the conventional full-QED calculation.  In the leading-soft eikonal sector, the soft-photon evolution is a unitary coherent-state displacement driven by the scattering current.  The equal-history identity of soft-photon influence functional $\mathcal F[J,J]=1$ exactly normalizes this soft evolution; together with the soft-photon theorem it removes the infrared-regulator-dependent leading-soft factor from each fixed-outcome inclusive probability.  Its $O(\alpha^2)$ expansion gives a compact check of that normalization.  I also derive an explicit leading-soft QED realization of scale-parametrized Lindblad evolution on a fixed hard-branch space. Tracing an infinitesimal soft-photon shell produces diagonal jump operators whose entries are fixed by the corresponding eikonal emission amplitudes.  In the closed-time-path representation, cross-branch Wightman contractions generate the jump term $L\rho L^\dagger$, while same-branch contractions generate the virtual subtraction $-\tfrac12\{L^\dagger L,\rho\}$ and the branch phase.  The finite-shell map is a completely positive unital Schur channel and, in the sharp scale-invariant leading-soft regime, a dephasing semigroup in $\log\Lambda$; more generally it is a completely-positive-divisible scale flow.  The resulting logarithmic visibility slope and monotonic purity loss are off-diagonal predictions of the reduced-state description.  The same controlled-displacement dilation gives the Sudakov no-emission probability, Poisson soft-photon multiplicities, and the $1/\omega$ bremsstrahlung number spectrum.
\end{abstract}

\section{Introduction}
\label{sec:intro}

Exclusive quantum electrodynamics (QED) amplitudes between Fock states contain infrared divergences from photons of arbitrarily small energy.  The no-soft-photon-emission probability for a charged particle changing its asymptotic velocity is driven to zero as the infrared regulator is removed. Measurements are nevertheless finite as they do not resolve arbitrarily soft radiation: the experimentally relevant quantity is an inclusive probability in which all unresolved photon states are summed. In this Bloch--Nordsieck (BN) approach the finiteness manifests only after virtual-loop and real-emission probabilities are evaluated separately and then combined~\cite{Bloch:1937pw,Yennie:1961ad}.

In this paper, I reformulate that calculation around the observables that are actually measured.  Photons below a resolution scale $\Lambda$ are treated as an environment and traced out.  The resolved charged particles, together with photons above $\Lambda$, form the hard system. The resulting reduced density matrix simultaneously contains inclusive probabilities on its diagonal and coherence between different hard outcomes off the diagonal.  This is the natural open-quantum-system (OQS)~\cite{OQS:2002BP} description of the soft sector of QED.  It is not a different dynamics or an alternative to QED: the partial trace is an exact rewriting of the same QED probability.  At a fixed perturbative order, with the same measured observable, diagrams, kinematics, and phase-space prescription, the OQS in Schwinger-Keldysh (SK)~\cite{Schwinger:1960qe,Keldysh:1964ud}, or closed-time-path (CTP)~\cite{Chou:1984es}, expansion reproduces both the infrared cancellation and the same surviving finite terms as the conventional full-QED calculation.

For a general finite-time process the combined hard and soft sectors evolve unitarily with $U(t_f,t_i)$.  For asymptotic scattering I use the interaction-picture limit $U_I(+\infty,-\infty)=S$.  Starting from a factorized preparation
$\rho^{\rm in}_{\tot}=\rho^{\rm in}_{\hard}\otimes|0\rangle\langle0|_\soft$, the reduced outgoing state is
\begin{equation}
 \rho^{\rm out}_{\hard}(\Lambda)
 =\Tr_\soft\!\left[S_\lambda\rho^{\rm in}_{\tot}S_\lambda^\dagger\right].
 \label{eq:reduced_intro}
\end{equation}
Here $\lambda$ is an auxiliary infrared regulator, while $\Lambda$ defines the resolved energy scale.  Physical fixed-resolution observables may depend on $\Lambda$ but must have a finite, regulator-independent limit as $\lambda$ is sent to zero.

Let $\Pi_f$ project onto a resolved hard measurement bin.  The label $f$ may denote any finite multiparticle hard configuration, including resolved photons; it is not restricted to a one-particle state.  For an incoming pure hard state $|i\rangle$,
\begin{equation}
 P_{f} (i; \Lambda)
 =\Tr_\hard\!\left(\Pi_f\rho^{\rm out}_{\hard}\right)
 =\sum_n\left|\langle f,n|S_\lambda|i,0\rangle\right|^2 .
 \label{eq:pf_intro}
\end{equation}
The sum runs over all unresolved soft-photon configurations.  For continuum scattering, $\Pi_f$ is understood as a finite kinematic bin or a wave-packet projector; the displayed amplitude sum then includes the corresponding sum or integral over hard states inside that bin.  The rank-one form is only a discrete shorthand.  The corresponding quantity is a differential probability or rate density, with the usual flux and phase-space factors supplied for a cross section.

Resolving the partial trace defines operators on the hard Hilbert space,
\begin{equation}
 K_n\equiv {}_\soft\!\langle n|S_\lambda|0\rangle_\soft,
 \qquad
 \rho^{\rm out}_{\hard}=\sum_nK_n\, \rho^{\rm in}_{\hard}\, K_n^\dagger .
 \label{eq:kraus_intro}
\end{equation}
For the outcome $f$, I define 
\begin{equation}
 F_f\equiv\sum_nK_n^\dagger\Pi_fK_n,
 \quad \mbox{such that} \quad P_{f} (i; \Lambda)=\langle i|F_f|i\rangle .
 \label{eq:povm_intro}
\end{equation}
The $F_f$ form the positive-operator-valued measure (POVM) induced on the incoming hard state.  Unitarity gives
\begin{equation}
 \sum_nK_n^\dagger K_n=\mathbb{I}_\hard,
 \qquad
 \sum_fF_f=\mathbb{I}_\hard,
 \qquad
 \sum_f P_{f} (i; \Lambda)=1 .
 \label{eq:global_norm_intro}
\end{equation}
This is the global normalization, not the Bloch--Nordsieck theorem.  The nontrivial statement is that each fixed-outcome probability in Eq.~\eqref{eq:pf_intro} has a finite $\lambda\to0$ limit.  That statement uses the soft-photon theorem~\cite{Weinberg:1965nx}: for a specified hard transition $i\to f$, the same eikonal current controls the infrared part of the virtual correction and the amplitudes for unresolved real emission.

The density-matrix structure requires an amplitude and its conjugate at the same time.  It is therefore naturally computed on a Schwinger--Keldysh (SK), or closed-time-path(CTP), contour~\cite{Schwinger:1960qe,Keldysh:1964ud,Chou:1984es}.  The photon propagator becomes
\begin{equation}
 \mathbf D=
 \begin{pmatrix}
 D_{++}&D_{+-}\\
 D_{-+}&D_{--}
 \end{pmatrix},
 \label{eq:ctp_matrix_intro}
\end{equation}
where same-branch lines generate virtual terms and cross-branch lines generate the cut, unresolved-real terms.  The four components satisfy
\begin{equation}
 D_{++}+D_{--}=D_{+-}+D_{-+}.
 \label{eq:ctp_identity_intro}
\end{equation}
At one loop in the eikonal approximation this identity places the virtual and real infrared kernels in a single expression.  At all orders in the abelian leading-soft eikonal sector the stronger statement is the equal-history normalization of the soft-photon's Feynman-Vernon influence functional, $\mathcal F[J,J]=1$.

The reduced-density-matrix viewpoint has been developed in studies of infrared quantum information, soft-radiation decoherence, and hard--soft entanglement~\cite{Carney:2017oxp,Gomez:2018,Semenoff:2019,Tomaras:2020bef,Irakleous:2021lgj}.  The role of SK observables in scattering on a deformed time contour has also been emphasized in recent works on asymptotic measurements~\cite{CaronHuot:2023vxl,CaronHuot:2025smatrix,Haehl:2024smatrix}.  On scale-parametrized Lindblad evolution, momentum-ordered soft evolution has been related to Lindblad/Banfi-Marchesini-Smye (BMS) dynamics~\cite{Neill:2018soft}, rapidity evolution of Color Glass Condensate density matrices has been cast in Lindblad form~\cite{Armesto:2019cgc,Li:2020jimwlk}, and renormalization-group evolution has been formulated as a quantum channel for spin decoherence induced by final-state radiation~\cite{Gu:2025rg}.  Detector resolution has also been interpreted as a coarse-graining scale of a reduced hard density matrix~\cite{Fukuyama:2026resolution}.  The specific result developed below is the microscopic leading-soft QED realization on a fixed hard-branch space: the shell jump operators are determined explicitly by the branch eikonal currents, and their Lindblad form is derived from the same coherent-state/CTP virtual--real structure that implements BN cancellation. My approach firmly puts resolution-dependent decoherence and scale-parametrized Lindblad evolution to the OQS foundation
. 

The paper is organized as follows.  Section~\ref{sec:setup} constructs the reduced hard channel, its Kraus operators, and the fixed-outcome quantum instrument.  Section~\ref{sec:ctp} gives the CTP formulation and the virtual/real dictionary.  Section~\ref{sec:oneloop} performs the one-loop cancellation and relates its common spectral kernel to the vacuum fluctuation--dissipation relation.  Section~\ref{sec:two_loop} gives the exact equal-history normalization of the leading-soft displacement channel and its compact $O(\alpha^2)$ check.  Section~\ref{sec:fgr} connects the construction with the Fermi Golden Rule and the optical theorem.  Section~\ref{sec:resolution_flow} constructs the fixed-branch soft-shell channel and derives its visibility and purity consequences.  Section~\ref{sec:coherent} collects the Sudakov, counting-statistics, spectrum, and energy consequences.  Section~\ref{sec:conclusions} concludes. Appendix A relates CTP propagator identities to vacuum fluctuation-dissipation theorem (FDT), B recalls soft exponent integral, C derives scale-local Lindblad generator, and D computes finite parts of CTP one-loop contributions.  

\section{Reduced Hard State and Quantum Instrument}
\label{sec:setup}

Consider a hard scattering process with characteristic scale $Q$ and massive charged external particles.  Choose $\lambda\ll\Lambda\ll Q$.  Photon energies and the sharp threshold $\Lambda$ are defined in a fixed detector frame.  The physical transverse photon modes are partitioned by energy,
\begin{equation}
 \mathcal H_{\rm photon}
 =\mathcal H_{\soft}(<\Lambda)\otimes\mathcal H_{\rm photon}(>\Lambda),
 \qquad
 \mathcal H_\soft(<\Lambda)=\mathcal F\{a^\dagger_{\mathbf k r}:|\mathbf k|<\Lambda\}.
 \label{eq:hilbert_split}
\end{equation}
I use covariant propagators in intermediate calculations and the conserved eikonal current to project onto physical polarizations.  The hard Hilbert space contains the charged particles and resolved photons.  Equation~\eqref{eq:hilbert_split} is the regulated perturbative Fock-space factorization of asymptotic transverse radiation modes used in inclusive calculations; it is not asserted to be a nonperturbative tensor factorization of the gauge-invariant charged Hilbert space.  Within this operational construction there is no inconsistency with Gauss' law: the unphysical polarizations decouple by current conservation, and the regulator is removed only after inclusive quantities are formed.  Changing $\Lambda$ changes which photon modes are discarded; it does not change the requirement that $\lambda$ disappear from physical fixed-resolution probabilities.

For the factorized preparation used throughout, Eq.~\eqref{eq:reduced_intro} is a completely positive trace-preserving (CPTP) map~\cite{Kraus:1983book,Nielsen:2000book}.  Its Kraus form is
\begin{equation}
 \mathcal E_{\Lambda,\lambda}[\rho]
 =\sum_{\mathbf n}K_{\mathbf n}\rho K_{\mathbf n}^\dagger,
 \qquad
 K_{\mathbf n}={}_\soft\!\langle\mathbf n|S_\lambda|0\rangle_\soft .
 \label{eq:kraus}
\end{equation}
The index $\mathbf n$ denotes a complete soft Fock configuration, not merely the total photon number.  Matrix elements are ordinary QED amplitudes,
\begin{equation}
 \langle f|K_{\mathbf n}|i\rangle
 =\langle f,\mathbf n|S_\lambda|i,0\rangle .
 \label{eq:kraus_matrix}
\end{equation}
Soft completeness and $S_\lambda^\dagger S_\lambda=\1$ give
\begin{equation}
 \sum_{\mathbf n}K_{\mathbf n}^\dagger K_{\mathbf n}=\1_\hard .
 \label{eq:kraus_complete}
\end{equation}
When contracted between hard states, the hard-output sum is implicit:
\begin{equation}
 \left\langle i'\left|\sum_{\mathbf n}K_{\mathbf n}^\dagger K_{\mathbf n}\right|i\right\rangle
 =\sum_{f,\mathbf n}\langle i',0|S_\lambda^\dagger|f,\mathbf n\rangle
 \langle f,\mathbf n|S_\lambda|i,0\rangle.
 \label{eq:kraus_complete_matrix}
\end{equation}

For a complete projective hard measurement $\{\Pi_f\}$, the selective measurement operation is
\begin{equation}
 \mathcal I_f[\rho]
 =\sum_{\mathbf n}\Pi_fK_{\mathbf n}\, \rho \, K_{\mathbf n}^\dagger\Pi_f.
 \label{eq:instrument}
\end{equation}
Its trace is the outcome probability,
\begin{equation}
 P_f(\rho)=\Tr\mathcal I_f[\rho]=\Tr(\rho F_f),
 \qquad
 F_f=\sum_{\mathbf n}K_{\mathbf n}^\dagger\Pi_fK_{\mathbf n}.
 \label{eq:povm}
\end{equation}
The unmeasured reduced channel is $\mathcal E_{\Lambda,\lambda}$.  If the projective outcome is measured but then ignored, the nonselective operation is instead
\begin{equation}
 \sum_f\mathcal I_f
 =\mathcal D_\Pi\circ\mathcal E_{\Lambda,\lambda},
 \qquad
 \mathcal D_\Pi[X]=\sum_f\Pi_fX\Pi_f,
 \label{eq:instrument_nonselective}
\end{equation}
which includes the dephasing caused by the measurement.  Thus $F_f$ acts on the incoming hard Hilbert space, whereas $\Pi_f$ selects an outgoing hard bin.  The distinction matters: $\sum_fF_f=\mathbb{I}$ is the full completeness, whereas infrared finiteness is a statement about each individual $P_f$.

For a fixed hard history $i\to f$, where $f$ may be a finite multiparticle configuration represented by a narrow hard bin, the soft-photon theorem~\cite{Weinberg:1965nx} gives, at leading order in every unresolved photon momentum,
\begin{equation}
 \langle f,\mathbf n|S_\lambda|i,0\rangle
 =\mathcal M^{\hard}_{fi}(\Lambda)\,s_{\mathbf n}[J_{fi};\lambda,\Lambda]
 +O(\omega/Q).
 \label{eq:soft_theorem}
\end{equation}
For charged external legs, with $Q_a$ measured in units of $e$,
\begin{equation}
 J_{fi}^{\mu}(k)=\sum_{a\in i,f}\eta_aQ_a\frac{p_a^\mu}{p_a\cdot k},
 \qquad
 \eta_a=\begin{cases}-1,&a\ \text{incoming},\\+1,&a\ \text{outgoing},\end{cases}
 \label{eq:eikonal_current}
\end{equation}
with charge conservation implying $k_\mu J_{fi}^{\mu}=0$.  For finite $\lambda$ and $\Lambda$, this current drives the free soft field into a multimode coherent state~\cite{Kibble:1968eqm}.  The soft amplitudes are matrix elements of a unitary displacement,
\begin{equation}
 s_{\mathbf n}[J_{fi}]=\langle\mathbf n|D(\alpha_{fi})|0\rangle,
 \qquad
 \sum_{\mathbf n}|s_{\mathbf n}[J_{fi}]|^2=1.
 \label{eq:soft_unitarity}
\end{equation}
Consequently the singular leading-soft factor cancels inside the fixed-outcome probability. This BN conclusion does not follow from the OQS definitions alone: it follows from the OQS/Kraus organization together with the soft-photon theorem and the unitary soft-photon displacement generated by the same fixed-hard-history on the two branches.  The soft-photon theorem is used here only to isolate the universal leading-soft terms.  Finite radiative terms, recoil effects, bin dependence, and subleading-soft contributions are not evaluated in Eq.~\eqref{eq:soft_theorem}; they are not altered by the OQS formulation.  When those terms are retained with the same diagrams, measurement operator, exact kinematics, and phase space, the OQS/CTP calculation gives exactly the same finite result as the conventional full-QED calculation.

\section{Closed-Time-Path Formulation}
\label{sec:ctp}

The full CTP generating functional is
\begin{equation}
 Z_{\rm full}[J_+,J_-]
 =\Tr\!\left[U[J_+]\rho^{\rm in}_{\tot}U^\dagger[J_-]\right],
 \label{eq:Zfull}
\end{equation}
while insertion of a hard projector defines
\begin{equation}
 Z_f[J_+,J_-]
 =\Tr\!\left[(\Pi_f\otimes\mathbb{I}_\soft)U[J_+]\rho^{\rm in}_{\tot}U^\dagger[J_-]\right].
 \label{eq:Zf}
\end{equation}
Equal sources give different quantities:
\begin{equation}
 Z_{\rm full}[J,J]=1,
 \qquad
 Z_f[J,J]=P_f.
 \label{eq:equal_source_distinction}
\end{equation}
After the hard history $i\to f$ has been selected by the soft theorem, the remaining soft influence functional is
\begin{equation}
 \mathcal F[J_+,J_-]
 =\Tr_\soft\!\left(U_\soft[J_+]\rho_\soft U_\soft^\dagger[J_-]\right),
 \qquad
 \mathcal F[J,J]=1.
 \label{eq:influence_def}
\end{equation}
The last identity is branchwise soft unitarity.  It should not be confused with the full trace identity or with the projected probability.

For the free photon field in the vacuum and the linearly coupled eikonal current, the Feynman--Vernon influence functional can be evaluated exactly~\cite{Feynman:1963fq}.  Introduce
$J_c=(J_++J_-)/2$ and $J_\Delta=(J_+-J_-)$.  The result can be written as 
\begin{equation}
 \mathcal F[J_+,J_-]
 =\exp\!\left[
   i e^2 J_\Delta\!\cdot D_R\!\cdot J_c
   -\frac{e^2}{2}J_\Delta\!\cdot D_H\!\cdot J_\Delta
 \right],
 \label{eq:IF_keldysh}
\end{equation}
where $D_R$ is the retarded propagator and
$D_H^{\mu\nu}(x-y)=\frac12\langle\{A^\mu(x),A^\nu(y)\}\rangle$ is the Hadamard correlator.  The retarded term carries the phase and causal response; the positive Hadamard term suppresses interference between unequal histories.  Equation~\eqref{eq:IF_keldysh} makes $\mathcal F[J,J]=1$ immediate because $J_\Delta=0$.

In the original $+/-$ basis the propagator is
\begin{equation}
 \mathbf D=
 \begin{pmatrix}
 D_{++}&D_{+-}\\
 D_{-+}&D_{--}
 \end{pmatrix},
 \quad
 \begin{aligned}
 D_{++}&=\langle TAA\rangle,&D_{--}&=\langle\bar TAA\rangle,\\
 D_{+-}&=\langle A_-A_+\rangle,&D_{-+}&=\langle A_+A_-\rangle.
 \end{aligned}
 \label{eq:ctp_components}
\end{equation}
Same-branch contractions are the virtual terms in the amplitude or its conjugate; cross-branch contractions are cut lines and therefore unresolved real radiation.

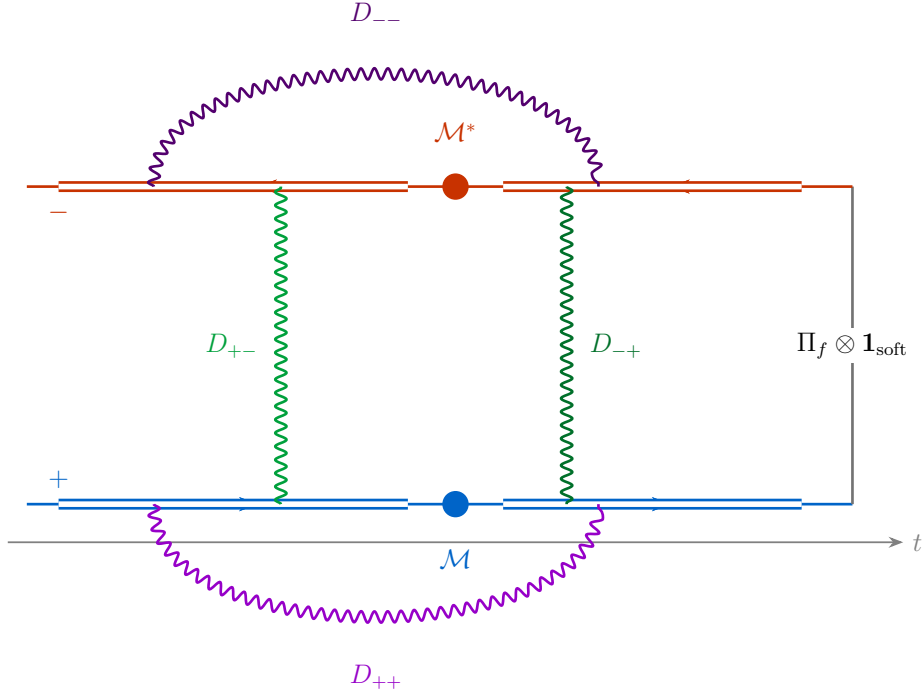
\begin{figure}[t]
\centering
\resizebox{0.82\textwidth}{!}{%
\begin{tikzpicture}[>=Stealth,thick]
  \draw[forwardcol,very thick,->-=0.28,->-=0.78] (-7,0)--(6,0);
  \node[forwardcol,above=3pt] at (-6.5,0) {$+$};
  \draw[black!55,very thick] (6,0)--(6,5);
  \node[fill=white,inner sep=2pt] at (6,2.5) {$\Pi_f\otimes\1_\soft$};
  \draw[backwardcol,very thick,->-=0.22,->-=0.72] (6,5)--(-7,5);
  \node[backwardcol,below=3pt] at (-6.5,5) {$-$};
  \draw[->,thick,gray] (-7.3,-0.6)--(6.8,-0.6) node[right,gray] {$t$};
  \draw[forwardcol,very thick,double,double distance=2.5pt] (-6.5,0)--(-1.0,0);
  \draw[forwardcol,very thick,double,double distance=2.5pt] (0.5,0)--(5.2,0);
  \filldraw[forwardcol] (-0.25,0) circle (5.5pt);
  \node[forwardcol,below=16pt] at (-0.25,0) {$\mathcal M$};
  \draw[backwardcol,very thick,double,double distance=2.5pt] (5.2,5)--(0.5,5);
  \draw[backwardcol,very thick,double,double distance=2.5pt] (-1.0,5)--(-6.5,5);
  \filldraw[backwardcol] (-0.25,5) circle (5.5pt);
  \node[backwardcol,above=16pt] at (-0.25,5) {$\mathcal M^*$};
  \draw[photoncol,very thick,decorate,decoration={snake,amplitude=2.5pt,segment length=6pt}]
    (-5.0,0)..controls(-5.0,-2.3)and(2.0,-2.3)..(2.0,0);
  \node[photoncol,below=2pt] at (-1.5,-2.3) {$D_{++}$};
  \draw[photoncol!55!black,very thick,decorate,decoration={snake,amplitude=2.5pt,segment length=6pt}]
    (-5.0,5)..controls(-5.0,7.3)and(2.0,7.3)..(2.0,5);
  \node[photoncol!55!black,above=2pt] at (-1.5,7.3) {$D_{--}$};
  \draw[cutcol,very thick,decorate,decoration={snake,amplitude=2.5pt,segment length=6pt}] (-3.0,0)--(-3.0,5);
  \node[cutcol,left=6pt] at (-3.0,2.5) {$D_{+-}$};
  \draw[cutcol!70!black,very thick,decorate,decoration={snake,amplitude=2.5pt,segment length=6pt}] (1.5,5)--(1.5,0);
  \node[cutcol!70!black,right=6pt] at (1.5,2.5) {$D_{-+}$};
\end{tikzpicture}%
}
\caption{Straight forward and backward CTP branches with the final-time insertion $\Pi_f\otimes\1_\soft$.  Same-branch photon lines generate virtual terms; cross-branch lines generate unresolved real-emission terms.  Removing $\Pi_f$ and summing over $f$ gives the full trace.}
\label{fig:ctp_main}
\end{figure}

I use the standard real-time convention in which the same-branch kernels include the Feynman-rule factor $i$, while the Wightman kernels are the real on-shell distributions displayed below.  Suppressing the common Lorentz and gauge tensor structure, vacuum photons obey
\begin{align}
 D_{++}(k)&=\frac{i}{k^2+i0},&
 D_{--}(k)&=\frac{-i}{k^2-i0},\\
 D_{+-}(k)&=2\pi\theta(-k^0)\delta(k^2),&
 D_{-+}(k)&=2\pi\theta(k^0)\delta(k^2).
 \label{eq:vacuum_ctp}
\end{align}
Sokhotski--Plemelj theorem on the real-time path gives
\begin{equation}
 D_{++}+D_{--}=D_{+-}+D_{-+}=2\pi\delta(k^2).
 \label{eq:ctp_identity}
\end{equation}
This is the one-line spectral identity behind the one-loop cancellation.  Appendix~\ref{app:ctp} records its retarded, advanced, and Keldysh forms.

\section{One-Loop Infrared Cancellation}
\label{sec:oneloop}

For clarity take one unit-charged particle deflected from $p_i$ to $p_f$; the multi-leg generalization follows from Eq.~\eqref{eq:eikonal_current}.  On shell,
\begin{equation}
 J^\mu(k)=\frac{p_f^\mu}{p_f\cdot k}-\frac{p_i^\mu}{p_i\cdot k},
 \qquad k^0=\omega=|\mathbf k|.
 \label{eq:single_current}
\end{equation}
Define the positive soft-radiation kernel
\begin{equation}
 N_{fi}(\lambda,\Lambda)
 \equiv e^2\sum_{r=1,2}\int_{\lambda<\omega<\Lambda}
 \frac{\dd^3{\bf k}}{(2\pi)^3\,2\omega_{\bf k}}
 \left|J(k)\cdot\epsilon_r^*(k)\right|^2\ge0.
 \label{eq:Ndef}
\end{equation}
It is simultaneously the mean number of unresolved photons in the coherent state and the exponent of the no-emission probability.  For a massive charge,
\begin{equation}
 N_{fi}=\frac{\alpha}{\pi}\,\mathcal A_{fi}\log\frac{\Lambda}{\lambda},
 \qquad
 \mathcal A_{fi}=\frac{1}{\beta_{if}}
 \log\frac{1+\beta_{if}}{1-\beta_{if}}-2\ge0,
 \label{eq:Nangular}
\end{equation}
where $\beta_{if}=\sqrt{1-m^4/(p_i\cdot p_f)^2}$; Appendix~\ref{app:sudakov} gives the angular integral.

The virtual soft correction is the zero-photon branch matrix element.  To first order in $e^2$,
\begin{equation}
 s_0[J_{fi}]
 =1-\frac12N_{fi}+i\phi_{fi}+O(e^4),
 \label{eq:s0_one_loop}
\end{equation}
where $\phi_{fi}$ is the principal-value Coulomb phase.  It may be infrared sensitive, but it cancels from the probability.  Hence the infrared-singular virtual contribution to the fixed-$f$ probability is
\begin{equation}
 \delta P^{(1),\mathrm{IR}}_{f,\virt}
 =-|\mathcal M^{\hard}_{fi}|^2N_{fi}.
 \label{eq:virtual_ir}
\end{equation}

The one-photon soft theorem gives
\begin{equation}
 \mathcal M_{fi+\gamma_r(k)}
 =e\,J(k)\cdot\epsilon_r^*(k)\,\mathcal M^{\hard}_{fi}+O(\omega^0).
 \label{eq:one_photon_amp}
\end{equation}
Squaring, summing over physical polarizations, and integrating over the unresolved phase space yields
\begin{equation}
 \delta P^{(1),\mathrm{IR}}_{f,\realp}
 =|\mathcal M^{\hard}_{fi}|^2N_{fi}.
 \label{eq:real_ir}
\end{equation}
Therefore
\begin{equation}
 \delta P^{(1),\mathrm{IR}}_{f,\virt}
 +\delta P^{(1),\mathrm{IR}}_{f,\realp}=0.
 \label{eq:bn_one_loop}
\end{equation}
Equation~\eqref{eq:bn_one_loop} displays the cancellation of the infrared-singular eikonal terms retained in the present calculation.  The remaining finite one-loop contributions are not evaluated in this section because it keeps only the leading-soft part.  Appendix~\ref{app:finite_ctp} evaluates them directly from the same CTP density-matrix expansion in two representative examples: the electron-line correction to Coulomb scattering and the final-state muon correction to $e^+e^-\to\mu^+\mu^-$.  In both cases the same-branch $++/--$ terms and the cross-branch $+-/-+$ terms reproduce the standard infrared-finite result without leaving the closed contour.

The CTP derivation shows why the two terms contain the same kernel.  The absorptive part of the time-ordered propagator is
\begin{equation}
 \frac{1}{k^2+i0}=\mathrm{P.V.}\frac1{k^2}-i\pi\delta(k^2),
 \label{eq:SP}
\end{equation}
so the virtual correction contains the same on-shell $\delta(k^2)$ measure as the Wightman propagators $D_{+-}$ and $D_{-+}$.  Inserting Eq.~\eqref{eq:ctp_identity} into the quadratic influence action gives, at equal histories,
\begin{equation}
 e^2J\cdot\left(D_{++}+D_{--}-D_{+-}-D_{-+}\right)\cdot J=0.
 \label{eq:one_loop_ctp_zero}
\end{equation}
Thus the loop and cut contributions are two placements of the same on-shell spectral kernel on the doubled contour.  In vacuum the shared on-shell spectral weight is equivalently encoded by the zero-temperature fluctuation--dissipation relation between the Hadamard fluctuations and the absorptive part of the retarded response.  That relation explains why the two kernels match; the fixed-outcome cancellation additionally uses the soft theorem to ensure that both terms are multiplied by the same hard transition amplitude.

\section{Equal-History Normalization at Second Order in the Leading-Soft Sector}
\label{sec:two_loop}

Within the abelian leading-soft approximation the hard trajectory is represented by its eikonal current and the soft photon field is free.  The resulting influence functional is evaluated exactly in Eq.~\eqref{eq:IF_keldysh}.  Equal histories have $J_\Delta=0$, so
\begin{equation}
 \mathcal F[J_{fi},J_{fi}]
 =\Tr_\soft\!\left(U_\soft[J_{fi}]\rho_\soft U_\soft^\dagger[J_{fi}]\right)=1.
 \label{eq:FJJ_exact}
\end{equation}
Equivalently, the soft evolution is a unitary displacement and its Fock probabilities are normalized.

The order-$\alpha^2$ check of this equal-history normalization is most transparent in terms of $N_{fi}=O(\alpha)$.  The probabilities for zero, one, and two unresolved photons are
\begin{align}
 P_0&=e^{-N_{fi}}
     =1-N_{fi}+\frac12N_{fi}^2+O(\alpha^3),\\
 P_1&=e^{-N_{fi}}N_{fi}
     =N_{fi}-N_{fi}^2+O(\alpha^3),\\
 P_2&=e^{-N_{fi}}\frac{N_{fi}^2}{2}
     =\frac12N_{fi}^2+O(\alpha^3).
 \label{eq:P012}
\end{align}
At $O(\alpha)$ the terms $-N_{fi}$ and $+N_{fi}$ are the virtual and one-real contributions.  At $O(\alpha^2)$ the virtual--virtual, virtual--real, and two-real terms cancel as
\begin{equation}
 \frac12N_{fi}^2-N_{fi}^2+\frac12N_{fi}^2=0,
 \qquad
 P_0+P_1+P_2=1+O(\alpha^3).
 \label{eq:two_loop_check}
\end{equation}
The same statement follows directly by expanding the coincident-history exponent: if the virtual and real one-loop contributions are $W_V=-N_{fi}$ and $W_R=+N_{fi}$, then the second-order coefficient is
$\frac12W_V^2+W_VW_R+\frac12W_R^2=\frac12(W_V+W_R)^2=0$.

This check is deliberately restricted to the $O(\alpha^2)$ coefficient of the leading-soft displacement sector.  It is not a substitute for evaluating the full set of two-loop QED diagrams, and it does not assert that arbitrary nested and crossed two-loop diagrams separately reduce to products of one-loop kernels.  A complete two-loop OQS/CTP calculation contains the same virtual graphs, cut graphs, recoil terms, subleading-soft terms, vacuum-polarization insertions, hard-loop regions, measurement function, and phase space as the conventional full-QED calculation, and therefore gives the same complete finite result.  Those contributions are simply not evaluated in the present leading-soft check.  What is established exactly here is the normalization of the leading-soft displacement sector.

For the fixed hard outcome, the leading-soft contribution is
\begin{equation}
 P_f^{\mathrm{LS}}(i;\Lambda)
 =|\mathcal M^{\hard}_{fi}(\Lambda)|^2
   \sum_{\mathbf n}|s_{\mathbf n}[J_{fi}]|^2
 =|\mathcal M^{\hard}_{fi}(\Lambda)|^2.
 \label{eq:fixed_f_all_orders}
\end{equation}
The unresolved sum removes the regulator-dependent leading-soft factor.  Equation~\eqref{eq:FJJ_exact} normalizes the soft channel, while Eq.~\eqref{eq:soft_theorem} identifies that channel in every unresolved sector of the fixed hard outcome.  Together they give the leading-soft BN cancellation.  The full fixed-outcome probability may still depend on the physical resolution $\Lambda$ and contains finite virtual and real terms, recoil effects, bin dependence, and subleading-soft contributions.  These contributions are not changed by the OQS/CTP organization: when computed for the same observable and at the same perturbative order, they are exactly the same as in the conventional full-QED result.

One might try to extend my approach to the more general Kinoshita-Lee-Nauenberg theorem, which incorporates an additional sum over initial degenerate states. Frye, Hannesdottier, Paul, Schwartz, and Yan showed that, once forward scattering is included, summing over either the initial or the final degenerate states can suffice for infrared finiteness~\cite{Frye:2018xjj}.

\section{Fermi's Golden Rule and the Optical Theorem}
\label{sec:fgr}

Use normalized asymptotic states and define invariant amplitudes by
\begin{equation}
 \langle X|iT|i\rangle
 =i(2\pi)^4\delta^{(4)}(P_X-P_i)\,\mathcal M_{Xi}.
 \label{eq:invariant_amp}
\end{equation}
For a resolved hard bin $f$, the soft-inclusive differential cross section is given by Fermi's Golden Rule~\cite{Dirac:1927df}:
\begin{equation}
 \dd\sigma^{\incl}_{i\to f}
 =\frac{1}{\mathcal F_i}
 \sum_{\mathbf n}
 |\mathcal M_{f\mathbf n,i}|^2\,
 \dd\Phi_{f\mathbf n},
 \label{eq:FGR}
\end{equation}
where $\mathcal F_i$ is the incoming flux and $\dd\Phi_{f\mathbf n}$ includes the full four-momentum-conserving delta function and the Lorentz-invariant phase space of the resolved hard configuration and unresolved photons.  For a decay, $1/\mathcal F_i$ is replaced by $1/(2E_i)$.  The sum over $\mathbf n$ is exactly the soft partial trace in Eq.~\eqref{eq:pf_intro}; the Fermi's Golden Rule is therefore the long-time rate form of the same reduced-state probability.

Summing also over all resolved hard outcomes gives the optical theorem,
\begin{equation}
 2\im\mathcal M_{ii}
 =\sum_{f,\mathbf n}\int\
 |\mathcal M_{f\mathbf n,i}|^2 \, \dd\Phi_{f\mathbf n},
 \label{eq:optical}
\end{equation}
with the same state-normalization convention as Eq.~\eqref{eq:invariant_amp}.  This is a global sum.  At fixed $f$, the soft theorem supplies the additional process-resolved information that turns the common on-shell spectral weight into the BN cancellation of Eqs.~\eqref{eq:virtual_ir} and \eqref{eq:real_ir}.  The CTP cut is its graphical implementation: the on-shell part of a same-branch propagator is matched to a cross-branch cut line carrying the same eikonal current.

\section{Soft-Shell Coarse Graining of Hard-Branch Coherence}
\label{sec:resolution_flow}

Soft-cloud suppression of off-diagonal hard-state elements and its dependence on an infrared resolution scale are known~\cite{Carney:2017oxp,Gomez:2018,Fukuyama:2026resolution}.  Scale-parametrized Lindblad evolution also has precedents in momentum-ordered soft dynamics, rapidity evolution, and renormalization-group evolution of collider density matrices~\cite{Neill:2018soft,Armesto:2019cgc,Li:2020jimwlk,Gu:2025rg}.  Here I isolate the narrower soft-QED result: an explicit fixed-branch shell channel whose jump operators are fixed by the eikonal emission amplitudes and whose Gorini-Kossakowski-Sudershan-Lindblad (GKSL)~\cite{GKSL:1976} structure follows directly from the CTP virtual--real decomposition.  Fix once and for all a finite branch Hilbert space
\begin{equation}
 \mathcal H_B=\mathrm{span}\{|a\rangle\},
 \label{eq:branch_space}
\end{equation}
whose orthonormal states are narrow, nonoverlapping wave packets of resolved hard configurations sharing the same incoming preparation.  Each outgoing branch may contain any finite collection of charged particles and resolved photons.  Over a soft window $\Lambda_1<\omega<\Lambda_2\ll Q$, hold this branch space, its hard QED amplitudes, and its measurement definition fixed, and trace only the additional photon shell.  Evaluate the leading eikonal current $J_a$ at the central momenta of each packet.

Let $\mathcal H_{21}$ be the Fock space of the photon shell $\Lambda_1<\omega<\Lambda_2$.  The leading-soft interaction defines the controlled displacement
\begin{equation}
 V_{21}
 =\sum_a|a\rangle\langle a|\otimes
 D_{21}(\alpha_a),
 \label{eq:shell_isometry}
\end{equation}
which is unitary on $\mathcal H_B\otimes\mathcal H_{21}$.  With the shell initially in the vacuum, the shell channel on the fixed branch space is
\begin{equation}
 \mathcal E_{21}[\rho]
 =\Tr_{21}\!\left[V_{21}
 (\rho\otimes|0\rangle\langle0|_{21})V_{21}^\dagger\right].
 \label{eq:shell_channel}
\end{equation}
Its matrix elements are
\begin{equation}
 (\mathcal E_{21}\rho)_{ab}
 =C^{21}_{ab}\rho_{ab},
 \qquad
 C^{21}_{ab}=\langle\alpha_b^{21}|\alpha_a^{21}\rangle .
 \label{eq:schur_flow}
\end{equation}
Thus no comparison of density matrices living in different hard Hilbert spaces is required: every shell map acts on the same prescribed branch space $\mathcal H_B$.

The coherent-state overlap gives
\begin{equation}
 C^{21}_{ab}
 =e^{i\Phi^{21}_{ab}}
 \exp\!\left[-\frac{e^2}{2}\sum_r
 \int_{\Lambda_1<\omega<\Lambda_2}
 \frac{\dd^3k}{(2\pi)^3\,2\omega}
 \left|(J_a-J_b)\cdot\epsilon_r^*\right|^2\right].
 \label{eq:resolution_overlap}
\end{equation}
Because the eikonal current scales as $1/\omega$, define
\begin{equation}
 \gamma_{ab}
 \equiv \frac{e^2}{4(2\pi)^3}\sum_r\int\dd\Omega_{\hat k}\,
 \left|\omega(J_a-J_b)\cdot\epsilon_r^*(\hat k)\right|^2\ge0,
 \qquad \gamma_{aa}=0.
 \label{eq:gamma_ab}
\end{equation}
Current conservation gives $k\cdot J_a=k\cdot J_b=0$.  Therefore the polarization shift
$\epsilon_r^\mu\to\epsilon_r^\mu+c_rk^\mu$ leaves the integrand unchanged, and $\gamma_{ab}$ is gauge independent.  In the sharp scale-invariant leading-soft window,
\begin{equation}
 |(\mathcal E_{21}\rho)_{ab}|
 =|\rho_{ab}|
 \left(\frac{\Lambda_2}{\Lambda_1}\right)^{-\gamma_{ab}}.
 \label{eq:power_coherence}
\end{equation}
Coherence is preserved precisely when the two hard branches have the same leading asymptotic soft current.

\begin{proposition}[Fixed-branch soft-shell channel]
For any finite branch space $\mathcal H_B$, Eq.~\eqref{eq:shell_channel} is a completely positive, trace-preserving, unital Schur channel.  Channels associated with adjacent, disjoint photon shells compose exactly.  With fixed eikonal currents and a sharp energy partition, the channel depends only on $\ell_2-\ell_1$, where $\ell=\log\Lambda$, and forms a homogeneous semigroup.  For a monotone smooth resolution profile with nonnegative shell weight, or for a scale-dependent kernel, the same construction gives a CP-divisible two-parameter flow with a scale-local generator.
\end{proposition}

\begin{proof}
The matrix $C^{21}$ is a Gram matrix of coherent states and is therefore positive semidefinite, with $C^{21}_{aa}=1$.  The Schur-product theorem makes $\rho\mapsto C^{21}\circ\rho$ completely positive, and the unit diagonal makes it trace preserving and unital.  Independent shell Fock spaces factorize, so $C^{31}_{ab}=C^{32}_{ab}C^{21}_{ab}$ and $\mathcal E_{31}=\mathcal E_{32}\circ\mathcal E_{21}$.  Scale homogeneity follows only when the leading-soft angular kernel is independent of $\ell$.
\end{proof}

For a possibly scale-dependent shell rate, the scale-local generator has diagonal GKSL form,
\begin{equation}
 \frac{\dd\rho}{\dd\ell}
 =-i[H_{\ell},\rho]
 +\sum_{r=1,2}\int\dd\Omega
 \left(L_{r\hat k}(\ell)\rho L^\dagger_{r\hat k}(\ell)
 -\frac12\left\{L^\dagger_{r\hat k}(\ell)L_{r\hat k}(\ell),\rho\right\}\right),
 \label{eq:resolution_gksl}
\end{equation}
where in the scale-invariant leading-soft limit
\begin{equation}
 L_{r\hat k}
 =\sqrt{\frac{e^2}{2(2\pi)^3}}
 \sum_a q_{ar}(\hat k)|a\rangle\langle a|,
 \qquad
 q_{ar}(\hat k)=\omega J_a\cdot\epsilon_r^*.
 \label{eq:resolution_lindblad_ops}
\end{equation}
The diagonal entry $\sqrt{\kappa}\,q_{ar}$ of $L_{r\hat k}$ on branch $a$ is the leading-soft shell-emission amplitude density; the corresponding infinitesimal Kraus amplitude is $\sqrt{\dd\ell}\,\sqrt{\kappa}\,q_{ar}$.  These jump operators are therefore fixed by the eikonal current rather than introduced phenomenologically.  Their CTP origin is direct: the cross-branch Wightman contraction generates $L_{r\hat k}\rho L_{r\hat k}^\dagger$, while the $++$ and $--$ contractions generate the virtual subtraction $-\tfrac12\{L_{r\hat k}^\dagger L_{r\hat k},\rho\}$ and the principal-value phase contained in $H_\ell$.  The detailed derivation from the infinitesimal-shell controlled displacement and the CTP kernel, including the normalization of the Lindblad operators and their Kraus expansion, is given in Appendix~\ref{app:gksl_derivation}.  The damping part gives
\begin{equation}
 \frac{\dd|\rho_{ab}|}{\dd\ell}=-\gamma_{ab}(\ell)|\rho_{ab}|,
 \qquad
 \frac{\dd}{\dd\ell}\Tr\rho^2
 =-2\sum_{a,b}\gamma_{ab}(\ell)|\rho_{ab}|^2\le0.
 \label{eq:resolution_master}
\end{equation}
For a finite branch space, unitality also implies nondecreasing von Neumann entropy.

For a two-branch interference observable, define
$\mathcal V_{ab}=2|\rho_{ab}|/(\rho_{aa}+\rho_{bb})$.  In the sharp scale-invariant regime, the fixed-branch shell map leaves the diagonal entries unchanged and therefore
\begin{equation}
 \frac{\dd\log\mathcal V_{ab}}{\dd\log\Lambda}
 =-\gamma_{ab},
 \qquad
 \frac{\mathcal V_{ab}(\Lambda_2)}{\mathcal V_{ab}(\Lambda_1)}
 =\left(\frac{\Lambda_2}{\Lambda_1}\right)^{-\gamma_{ab}},
 \label{eq:visibility_prediction}
\end{equation}
provided the whole interval remains inside the controlled soft window.  This is a prediction for repeated coarse grainings of the same hard-branch ensemble.  It is not a claim that a complete experimental cross section has threshold-independent diagonal entries: changing an actual event definition can also change the hard QED amplitudes entering the selected bin, resolved-photon classification, and bin migration.  For any fixed event definition, however, an OQS/CTP calculation and a conventional full-QED calculation give the same diagonal probabilities.

For an initially pure two-branch state
$\sqrt p\,|a\rangle+\sqrt{1-p}\,|b\rangle$, define
$c_{ab}(\Lambda_2,\Lambda_1)=C^{21}_{ab}$ after removing its phase.  The reduced-state eigenvalues are
\begin{equation}
 \lambda_{\pm}=\frac12\left[1\pm
 \sqrt{1-4p(1-p)\left(1-|c_{ab}|^2\right)}\right],
 \qquad
 S_{\hard}=-\sum_{\sigma=\pm}\lambda_\sigma\log\lambda_\sigma.
 \label{eq:two_branch_entropy}
\end{equation}
For $p=1/2$, $\lambda_\pm=(1\pm|c_{ab}|)/2$, so the visibility slope fixes the entropy growth of the branch state.  The resolution dependence of soft-cloud overlap is not itself new; the result here is its explicit realization as a fixed-space Schur channel with a scale-local GKSL generator.  Alternative asymptotic dressed-state constructions can assign the leading-soft coherence differently~\cite{DeLisle:2022}; the present statement applies to the regulated inclusive Fock-space coarse graining defined in Sec.~\ref{sec:setup}.

\section{Coherent-State Consequences}
\label{sec:coherent}

The soft evolution for a fixed hard history is the multimode displacement
\begin{equation}
 D(\alpha_{fi})
 =\exp\!\left[\sum_r\int\frac{\dd^3k}{(2\pi)^3}
 \left(\alpha_{fi,r}(k)a_r^\dagger(k)-\alpha_{fi,r}^*(k)a_r(k)\right)\right],
 \label{eq:displacement}
\end{equation}
with mode function chosen so that
\begin{equation}
 \|\alpha_{fi}\|^2
 =e^2\sum_r\int_{\lambda<\omega<\Lambda}
 \frac{\dd^3k}{(2\pi)^3\,2\omega}
 |J_{fi}\cdot\epsilon_r^*|^2=N_{fi}.
 \label{eq:alpha_norm}
\end{equation}
The zero-emission Kraus matrix element is
\begin{equation}
 \langle f|K_0|i\rangle
 =\mathcal M_{fi}^{\hard}\,e^{-N_{fi}/2+i\phi_{fi}},
 \qquad
 \frac{|\langle f|K_0|i\rangle|^2}{|\mathcal M_{fi}^{\hard}|^2}=e^{-N_{fi}}.
 \label{eq:sudakov}
\end{equation}
This is the soft no-emission (Sudakov) probability on the specified hard branch~\cite{Sudakov:1954sw}; it is a matrix element, not an operator norm.

The complementary soft output of the same controlled-displacement dilation has Poisson statistics,
\begin{equation}
 P(N_\gamma=n)=e^{-N_{fi}}\frac{N_{fi}^n}{n!},
 \qquad
 G(z)=\sum_{n=0}^{\infty}z^nP(n)=e^{N_{fi}(z-1)}.
 \label{eq:poisson}
\end{equation}
Hence the mean and variance of the unresolved photon number are both $N_{fi}$.  Independent frequency shells factorize in the same dilation that produced the channel composition in Sec.~\ref{sec:resolution_flow}.

Writing
\begin{equation}
 \mathcal A_{fi}=\int\frac{\dd\Omega}{4\pi}
 \left[-\omega^2J_{fi}^2(\omega,\hat k)\right],
 \label{eq:Aangular}
\end{equation}
Eq.~\eqref{eq:Ndef} gives the number spectrum
\begin{equation}
 \frac{\dd\bar n}{\dd\omega}
 =\frac{\alpha}{\pi}\frac{\mathcal A_{fi}}{\omega}.
 \label{eq:number_spectrum}
\end{equation}
The photon number diverges logarithmically as $\lambda\to0$, but the radiated energy in the unresolved band is finite,
\begin{equation}
 \langle E_\soft\rangle
 =\int_\lambda^\Lambda\dd\omega\,\omega\frac{\dd\bar n}{\dd\omega}
 =\frac{\alpha}{\pi}\mathcal A_{fi}(\Lambda-\lambda).
 \label{eq:soft_energy}
\end{equation}
The infrared singularity is therefore a divergence of multiplicity, not energy.  The same mode amplitudes determine the branch-overlap distance in Eq.~\eqref{eq:resolution_overlap}; photon counting in the complementary soft output and hard-branch dephasing are two observables of one hard--soft unitary dilation.

\section{Conclusions}
\label{sec:conclusions}

Soft QED is naturally organized as an open quantum system when photons below a finite resolution are unobserved.  The reduced hard density matrix is the primary object: its diagonal elements are fixed-outcome inclusive probabilities and its off-diagonal elements retain coherence information absent from an ordinary cross-section calculation.  The hard outcome may be any finite resolved multiparticle configuration.  Kraus operators label unresolved photon states, while the effects $F_f$ encode resolved hard measurements.

The CTP formulation places the amplitude and its conjugate on one doubled contour.  At one loop, the absorptive part of a same-branch propagator and the cross-branch Wightman propagators carry the same on-shell eikonal kernel, giving equal and opposite infrared terms at fixed hard outcome.  In the abelian leading-soft eikonal sector, $\mathcal F[J,J]=1$ exactly normalizes the soft displacement channel.  Combined with the soft theorem, this gives the all-orders cancellation of the regulator-dependent leading-soft factor.  The order-$\alpha^2$ identity $\frac12N^2-N^2+\frac12N^2=0$ checks that leading-soft normalization; it does not replace the evaluation of the remaining two-loop terms.  When the full set of terms is included with the same observable, the OQS/CTP expansion is exactly the same QED calculation and gives the same complete infrared-finite result as the conventional full-QED organization.

The additional channel result is formulated on a fixed finite space of hard wave-packet branches.  Tracing an independent photon shell gives a CPTP unital Schur map.  Sharp shells with fixed eikonal currents form a homogeneous dephasing semigroup in $\log\Lambda$; more general resolution profiles give a CP-divisible scale flow.  Scale-parametrized Lindblad evolution itself has precedents in high-energy theory.  The specific contribution here is the explicit leading-soft QED generator: the diagonal entries of its jump operators are fixed by the eikonal shell-emission amplitudes, and the jump term and virtual subtraction arise respectively from cross-branch and same-branch CTP contractions.  The resulting current-difference kernel determines a gauge-independent logarithmic visibility slope and monotonic purity loss.

The construction uses massive charges, factorized initial hard--soft states, a vacuum soft environment, regulated perturbative asymptotic Fock space, and the leading-soft eikonal approximation.  It does not assert a nonperturbative hard--soft factorization of the physical charged Hilbert space, nor does it address how an asymptotic dressed-state formulation redistributes leading-soft coherence.  Within its stated domain, it gives a compact OQS account of BN cancellation, Sudakov suppression, soft-photon counting, and soft-shell decoherence.

\section*{Acknowledgments}
I acknowledge feedbacks from participants of workshops at the Institute for Nuclear Theory (December 2025, USA) and at the Mainz Institute for Theoretical Physics (April 2026, Germany). Part of this work was performed while I was visiting the Institute for Pure and Applied Mathematics (IPAM, USA) and the Simons Institute for the Theory of Computing (SIfTC, USA). This work was supported in part by the DFG Cluster of Excellence PRISMA+ (Project ID 39083469), by the U.S. National Science Foundation through IPAM and SIfTC, by the Simons Foundation through SIfTC, and by the National Research Foundation of Korea (NRF) (RS-2021-NR060112) and Kwangwoon University. 
\appendix

\section{CTP Propagator Identities and Vacuum Fluctuation Dissipation Relation}
\label{app:ctp}

Define
\begin{align}
 D_R&=D_{++}-D_{+-}=D_{-+}-D_{--},\\
 D_A&=D_{++}-D_{-+}=D_{+-}-D_{--},\\
 D_\rho&=D_{-+}-D_{+-}=D_R-D_A,\\
 D_K&=D_{-+}+D_{+-}.
\end{align}
The contour identity is
\begin{equation}
 D_{++}+D_{--}=D_{+-}+D_{-+}=D_K.
\end{equation}
For a vacuum massless photon,
\begin{equation}
 D_\rho(k)=2\pi\,\mathrm{sgn}(k^0)\delta(k^2),
 \qquad
 D_K(k)=2\pi\delta(k^2)
 =\mathrm{sgn}(k^0)D_\rho(k).
\end{equation}
The first equality is the contour-ordering identity.  The further equilibrium relation
$D_K(k)=\mathrm{sgn}(k^0)D_\rho(k)$ is the zero-temperature fluctuation--dissipation relation (FDR); at temperature $T$ it becomes $D_K(k)=\coth(k^0/2T)D_\rho(k)$.  In the influence functional the Hadamard kernel is $D_H=D_K/2$.  FDR explains why the noise responsible for real radiation and the absorptive response responsible for virtual suppression contain the same spectral density.  Fixed-outcome BN cancellation additionally requires the soft theorem and the same hard history on both branches.

\section{Angular Integral for the Soft Exponent}
\label{app:sudakov}

For one massive charge deflected from $p_i$ to $p_f$,
\begin{equation}
 N_{fi}=e^2\int_{\lambda<\omega<\Lambda}
 \frac{\dd^3k}{(2\pi)^3\,2\omega}\left[-J^2(k)\right]
 =\frac{\alpha}{\pi}\mathcal A_{fi}\log\frac{\Lambda}{\lambda},
 \label{eq:N_app}
\end{equation}
where
\begin{equation}
 \mathcal A_{fi}=\int\frac{\dd\Omega}{4\pi}\left[-\omega^2J^2(k)\right].
 \label{eq:A_app}
\end{equation}
Choose a frame in which
$p_i=E(1,0,0,\beta)$ and
$p_f=E(1,\beta\sin\theta,0,\beta\cos\theta)$.  Then
\begin{equation}
 -\omega^2J^2
 =\frac{2(1-\beta^2\cos\theta)}
 {(1-\beta\hat v_i\cdot\hat k)(1-\beta\hat v_f\cdot\hat k)}
 -\frac{1-\beta^2}{(1-\beta\hat v_i\cdot\hat k)^2}
 -\frac{1-\beta^2}{(1-\beta\hat v_f\cdot\hat k)^2}.
\end{equation}
The two self terms each integrate to $-1$.  Feynman parametrizing the cross term gives the covariant result
\begin{equation}
 \mathcal A_{fi}
 =\frac1{\beta_{if}}\log\frac{1+\beta_{if}}{1-\beta_{if}}-2,
 \qquad
 \beta_{if}=\sqrt{1-\frac{m^4}{(p_i\cdot p_f)^2}}.
 \label{eq:A_final}
\end{equation}
It is nonnegative and vanishes when $p_f=p_i$.  In the ultrarelativistic but massive regime $p_i\cdot p_f\gg m^2$,
\begin{equation}
 \mathcal A_{fi}
 =2\left[\log\frac{2p_i\cdot p_f}{m^2}-1\right]+O\!\left(\frac{m^4}{(p_i\cdot p_f)^2}\log\frac{p_i\cdot p_f}{m^2}\right),
\end{equation}
so
\begin{equation}
 N_{fi}\simeq\frac{2\alpha}{\pi}
 \left[\log\frac{2p_i\cdot p_f}{m^2}-1\right]
 \log\frac{\Lambda}{\lambda}.
\end{equation}

The mass keeps the angular integral finite; the remaining divergence is the soft logarithm.

\section{Derivation of the Scale-Local GKSL Generator}
\label{app:gksl_derivation}

This appendix derives Eq.~\eqref{eq:resolution_gksl} directly from the controlled displacement of an infinitesimal photon shell.  Let \(\ell=\log\Lambda\) and trace the shell \(\ell<\log\omega<\ell+\dd\ell\), while the finite hard-branch space \(\mathcal H_B=\mathrm{span}\{|a\rangle\}\) is held fixed.  For branch \(a\), define
\begin{equation}
 q_{ar}(\hat k,\ell)
 =\omega J_a(\omega,\hat k;\ell)\cdot\epsilon_r^*(\hat k),
 \qquad
 \kappa\equiv\frac{e^2}{2(2\pi)^3}.
 \label{eq:gksl_q_kappa}
\end{equation}
In the scale-invariant leading-soft limit \(q_{ar}\) is independent of \(\ell\); the derivation below also allows a scale-dependent current or a smooth nonnegative shell weight, which can be absorbed into \(q_{ar}(\hat k,\ell)\).

The displacement profile carried by the infinitesimal shell has norm and mutual inner product
\begin{align}
 \|\dd\alpha_a\|^2
 &=\kappa\,\dd\ell\sum_r\int\dd\Omega\,
 |q_{ar}(\hat k,\ell)|^2,
 \label{eq:gksl_alpha_norm}\\
 (\dd\alpha_b,\dd\alpha_a)
 &=\kappa\,\dd\ell\sum_r\int\dd\Omega\,
 q_{br}^*(\hat k,\ell)q_{ar}(\hat k,\ell).
 \label{eq:gksl_alpha_inner}
\end{align}
The coefficient follows directly from the on-shell measure:
\begin{equation}
 e^2\frac{\dd^3k}{(2\pi)^3 2\omega}
 =\frac{e^2}{2(2\pi)^3}\frac{\dd\omega}{\omega}\dd\Omega\,\omega^2,
 \qquad \frac{\dd\omega}{\omega}=\dd\ell,
 \label{eq:gksl_measure}
\end{equation}
which, together with \(J_a\sim\omega^{-1}\), gives Eqs.~\eqref{eq:gksl_alpha_norm} and \eqref{eq:gksl_alpha_inner}.

On the branch basis, the infinitesimal-shell isometry acts as
\begin{equation}
 V_{\ell+\dd\ell,\ell}
 |a\rangle|0\rangle
 =e^{-ih_a(\ell)\dd\ell}|a\rangle|\dd\alpha_a\rangle,
 \label{eq:gksl_inf_isometry}
\end{equation}
where the real diagonal quantity \(h_a(\ell)\) contains the principal-value phase generated in that shell.  Since coherent states obey
\begin{equation}
 \langle\beta|\alpha\rangle
 =\exp\!\left[-\frac12\|\alpha\|^2
 -\frac12\|\beta\|^2+(\beta,\alpha)\right],
 \label{eq:gksl_coherent_overlap}
\end{equation}
the shell map gives
\begin{align}
 (\mathcal E_{\ell+\dd\ell,\ell}\rho)_{ab}
 &=e^{-i[h_a(\ell)-h_b(\ell)]\dd\ell}
 \langle\dd\alpha_b|\dd\alpha_a\rangle\rho_{ab}
 \nonumber\\
 &=\rho_{ab}+\dd\ell\Bigg\{-i[h_a-h_b]
 +\kappa\sum_r\int\dd\Omega\,
 q_{ar}q_{br}^* \nonumber\\
 &\hspace{31mm}
 -\frac{\kappa}{2}\sum_r\int\dd\Omega\,
 \bigl(|q_{ar}|^2+|q_{br}|^2\bigr)\Bigg\}\rho_{ab}
 +O(\dd\ell^2).
 \label{eq:gksl_matrix_derivation}
\end{align}
All quantities in the second line are evaluated at \(\ell\).

The three terms in Eq.~\eqref{eq:gksl_matrix_derivation} are exactly the infinitesimal-shell CTP contributions.  The cross-branch Wightman contraction---$D_{-+}$ for one momentum routing, equivalently $D_{+-}$ after reversing the shell momentum---places the shell photon on shell and gives
\begin{equation}
 \left.\frac{\dd\rho_{ab}}{\dd\ell}\right|_{\rm cross}
 =\kappa\sum_r\int\dd\Omega\,
 q_{ar}q_{br}^*\rho_{ab}.
 \label{eq:gksl_ctp_cross}
\end{equation}
The $++$ insertion on the ket branch and the $--$ insertion on the bra branch give
\begin{equation}
 \left.\frac{\dd\rho_{ab}}{\dd\ell}\right|_{\rm same}
 =-\frac{\kappa}{2}\sum_r\int\dd\Omega\,
 \bigl(|q_{ar}|^2+|q_{br}|^2\bigr)\rho_{ab}
 -i(h_a-h_b)\rho_{ab}.
 \label{eq:gksl_ctp_same}
\end{equation}
The real parts of the same-branch kernels are the virtual normalization terms accompanying the cross-branch cut, while their principal-value parts produce the diagonal phase.  Thus Eqs.~\eqref{eq:gksl_ctp_cross} and \eqref{eq:gksl_ctp_same} are the scale-local version of the same CTP real--virtual pairing used in Sec.~\ref{sec:oneloop}.

Now define the diagonal operators
\begin{equation}
 H_\ell=\sum_a h_a(\ell)|a\rangle\langle a|,
 \qquad
 L_{r\hat k}(\ell)
 =\sqrt{\kappa}\sum_a q_{ar}(\hat k,\ell)
 |a\rangle\langle a|.
 \label{eq:gksl_appendix_operators}
\end{equation}
Their matrix elements satisfy
\begin{align}
 [L_{r\hat k}\rho L_{r\hat k}^\dagger]_{ab}
 &=\kappa q_{ar}q_{br}^*\rho_{ab},\\
 \left[\frac12\{L_{r\hat k}^\dagger L_{r\hat k},\rho\}\right]_{ab}
 &=\frac{\kappa}{2}
 (|q_{ar}|^2+|q_{br}|^2)\rho_{ab}.
 \label{eq:gksl_dissipator_elements}
\end{align}
Substituting Eq.~\eqref{eq:gksl_dissipator_elements} into Eq.~\eqref{eq:gksl_matrix_derivation} yields the operator equation
\begin{equation}
 \frac{\dd\rho}{\dd\ell}
 =-i[H_\ell,\rho]
 +\sum_r\int\dd\Omega\left(
 L_{r\hat k}\rho L_{r\hat k}^\dagger
 -\frac12\{L_{r\hat k}^\dagger L_{r\hat k},\rho\}
 \right),
 \label{eq:gksl_appendix_result}
\end{equation}
which is Eq.~\eqref{eq:resolution_gksl}.  Equivalently, the infinitesimal-shell Kraus operators may be chosen as
\begin{align}
 K_0
 &=\mathbf1-\dd\ell\left(
 iH_\ell+\frac12\sum_r\int\dd\Omega\,
 L_{r\hat k}^\dagger L_{r\hat k}
 \right)+O(\dd\ell^2),
 \label{eq:gksl_K0}\\
 K_{r\hat k}
 &=\sqrt{\dd\ell}\,L_{r\hat k}+O(\dd\ell^{3/2}).
 \label{eq:gksl_K1}
\end{align}
They obey
\begin{equation}
 K_0^\dagger K_0
 +\sum_r\int\dd\Omega\,
 K_{r\hat k}^\dagger K_{r\hat k}
 =\mathbf1+O(\dd\ell^2),
 \label{eq:gksl_kraus_completeness}
\end{equation}
and their Kraus sum reproduces Eq.~\eqref{eq:gksl_appendix_result}.  Two-or-more-photon Kraus sectors in one infinitesimal shell begin at \(O(\dd\ell^2)\) and therefore do not enter the scale-local generator.

Finally,
\begin{equation}
 \operatorname{Re}\left[
 q_{ar}q_{br}^*-\frac12|q_{ar}|^2-\frac12|q_{br}|^2
 \right]
 =-\frac12|q_{ar}-q_{br}|^2.
 \label{eq:gksl_real_identity}
\end{equation}
Hence
\begin{equation}
 \frac{\dd|\rho_{ab}|}{\dd\ell}
 =-\frac{\kappa}{2}\sum_r\int\dd\Omega\,
 |q_{ar}-q_{br}|^2|\rho_{ab}|
 =-\gamma_{ab}(\ell)|\rho_{ab}|,
 \label{eq:gksl_gamma_match}
\end{equation}
with \(\kappa/2=e^2/[4(2\pi)^3]\), exactly matching Eq.~\eqref{eq:gamma_ab}.  For \(a=b\), the dissipator vanishes.  Because all \(L_{r\hat k}\) are diagonal and therefore normal, \(\mathcal L_\ell[\mathbf1]=0\); the scale-local map is unital as well as trace preserving.  Nonnegative shell weights make every instantaneous rate nonnegative, which is the CP-divisibility condition used in Sec.~\ref{sec:resolution_flow}.

\section{Finite CTP One-Loop Examples}
\label{app:finite_ctp}

This appendix reinstates two finite one-loop examples that were displayed in an earlier version of the manuscript, but now derives them directly in the closed-time-path expansion.  The examples are deliberately restricted to a gauge-invariant charged-line subset: (i) the electron-line correction in scattering from a nonradiating external Coulomb source, and (ii) the final-state muon correction in $e^+e^-\to\mu^+\mu^-$.  Vacuum polarization, radiation from the source, initial-state radiation in the second example, and power-suppressed Pauli-form-factor terms are not part of these showcase subsets.  The photon mass $\lambda$ regulates the infrared region, and the unresolved photon energy is restricted to $k^0<\Delta$ in the symmetric frame specified below.

To make the normalization explicit, write each one-loop CTP-sector contribution to the diagonal hard probability as
\begin{equation}
 \frac{p_{f,X}^{(1)}}{p_f^{(0)}}
 \equiv \frac{\alpha}{\pi}\,\mathcal B_X,
 \qquad X=\mathrm V,\mathrm R,\mathrm{V+R},
 \qquad \alpha\equiv\frac{e^2}{4\pi}.
 \label{eq:B_coefficient_definition}
\end{equation}
Thus $\mathcal B_X$ is a dimensionless coefficient with the universal coupling factor $\alpha/\pi$ stripped off.  In the photon-mass scheme used below, $[\mathcal B_X]_{\rm fin}$ denotes the remainder after the explicit $\log(m/\lambda)$ term has been separated.  Only the sum of virtual and real sectors is regulator independent.

\subsection{Same-branch virtual term}

Let $\rho_0$ denote the Born hard density matrix.  A one-loop vertex insertion on the forward branch gives $\delta\rho_{++}=\Gamma_{++}^{(1)}\rho_0$, while the backward-branch insertion gives $\delta\rho_{--}=\rho_0\Gamma_{--}^{(1)}$ with $\Gamma_{--}^{(1)}=\Gamma_{++}^{(1)\dagger}$.  Thus the virtual coefficient in Eq.~\eqref{eq:B_coefficient_definition} is obtained inside the CTP kernel as
\begin{equation}
 \frac{\alpha}{\pi}\,\mathcal B_{\mathrm V}
 =\frac{\Tr\!\left[\Pi_f(\Gamma_{++}^{(1)}\rho_0+\rho_0\Gamma_{--}^{(1)})\right]}
 {\Tr(\Pi_f\rho_0)}
 =2\,\mathrm{Re}\,F_{1,++}^{(1)}+O(m^2/Q^2),
 \label{eq:ctp_virtual_definition}
\end{equation}
where the last form applies to the high-energy Dirac form factor.  In Feynman gauge the renormalized $+$-branch vertex kernel is
\begin{align}
 \Gamma_{++,R}^{\mu(1)}(p',p)
 &=(-ie)^2\!\int\!\frac{\dd^4\ell}{(2\pi)^4}\,
 \gamma^\alpha S_{++}(p'-\ell)\gamma^\mu
 S_{++}(p-\ell)\gamma^\beta
 D^{(\lambda)}_{++,\alpha\beta}(\ell)
 +\delta\Gamma_{++,\mathrm{OS}}^\mu,\label{eq:ctp_vertex_loop}\\
 S_{++}(r)&=\frac{i(\slashed r+m)}{r^2-m^2+i0},\qquad
 D^{(\lambda)}_{++,\alpha\beta}(\ell)
 =\frac{-ig_{\alpha\beta}}{\ell^2-\lambda^2+i0}.
\end{align}
No in--out probability is introduced here: Eq.~\eqref{eq:ctp_vertex_loop} is the loop insertion on the ket branch of the reduced-density-matrix contour, and the $--$ branch supplies its conjugate.  Feynman parameterization of this CTP $++$ kernel followed by the on-shell subtraction at $q^2=0$ gives, for $L\equiv\log(Q^2/m^2)$ and $Q^2\gg m^2$,
\begin{equation}
 \mathcal B_{\mathrm V}^{\mathrm{SL}}
 =-\frac12L^2+\frac32L
 -2(L-1)\log\frac{m}{\lambda}
 -2+\frac{\pi^2}{6}
 +O\!\left(\frac{m^2}{Q^2}\right)
 \label{eq:virtual_spacelike_ctp}
\end{equation}
for spacelike transfer $q^2=-Q^2$.  For timelike transfer $q^2=s=Q^2$, the $+i0$ prescription of the same $++$ kernel gives $L\to L-i\pi$.  Taking the real part required by Eq.~\eqref{eq:ctp_virtual_definition},
\begin{equation}
 \mathcal B_{\mathrm V}^{\mathrm{TL}}
 =-\frac12L^2+\frac32L
 -2(L-1)\log\frac{m}{\lambda}
 -2+\frac{2\pi^2}{3}
 +O\!\left(\frac{m^2}{Q^2}\right).
 \label{eq:virtual_timelike_ctp}
\end{equation}
The difference between Eqs.~\eqref{eq:virtual_spacelike_ctp} and \eqref{eq:virtual_timelike_ctp} is the finite $+\pi^2/2$ generated by the timelike analytic continuation.  The imaginary Coulomb phase in the $++$ kernel cancels against the $--$ branch in the diagonal probability.

\subsection{Cross-branch real term}

The unresolved one-photon contribution is a cross-branch contraction.  Route the positive-energy momentum from the $+$ branch to the $-$ branch:
\begin{equation}
 \frac{\alpha}{\pi}\,\mathcal B_{\mathrm R}
 =-e^2\!\int_{k^0<\Delta}\!\frac{\dd^4k}{(2\pi)^4}\,
 D_{-+}^{(\lambda)}(k)\,J_\mu(k)J^\mu(k),\qquad
 D_{-+}^{(\lambda)}(k)=2\pi\theta(k^0)\delta(k^2-\lambda^2).
 \label{eq:ctp_real_cut}
\end{equation}
The opposite routing is represented by $D_{+-}$ and gives the same physical cut after the loop momentum is reversed; it is not an additional photon state.  Equation~\eqref{eq:ctp_real_cut} is therefore precisely the one-photon phase-space term generated by the cross-branch sector of the density matrix.  The factor in front is consistent with Eq.~\eqref{eq:B_coefficient_definition}: inserting $e^2=4\pi\alpha$ and the on-shell Wightman measure leaves exactly the overall factor $\alpha/\pi$ multiplying the dimensionless angular--energy integral below.

For both examples choose a symmetric frame in which the two charged legs entering the eikonal current are
\begin{equation}
 p_1=E(1,0,0,\beta),\qquad p_2=E(1,0,0,-\beta),\qquad
 J^\mu(k)=\frac{p_1^\mu}{p_1\cdot k}-\frac{p_2^\mu}{p_2\cdot k}.
 \label{eq:symmetric_current}
\end{equation}
For Coulomb scattering this is the Breit frame of the electron line; for pair production it is the $\mu^+\mu^-$ center-of-mass frame.  Write $k^0=\lambda\cosh u$, $|\mathbf k|=\lambda\sinh u$, and $u_\Delta=\operatorname{arcosh}(\Delta/\lambda)$.  Performing the angular integral in Eq.~\eqref{eq:ctp_real_cut} gives the one-dimensional CTP cut coefficient
\begin{align}
 \mathcal B_{\mathrm R}
 &=\int_0^{u_\Delta}\!\dd u\;2\tanh^2u
 \left[
 -\frac{1-\beta^2}{1-\beta^2\tanh^2u}
 +\frac{1+\beta^2}{\beta\tanh u}
 \operatorname{artanh}(\beta\tanh u)
 \right].
 \label{eq:ctp_real_u}
\end{align}
This formula follows only from the Wightman kernel on the crossed branches.  Expanding its exact endpoint value for $Q^2\gg m^2$ gives
\begin{equation}
 \mathcal B_{\mathrm R}
 =2(L-1)\log\frac{2\Delta}{\lambda}
 -\frac12L^2+L-\frac{\pi^2}{3}
 +O\!\left(\frac{m^2}{Q^2}\log\frac{Q^2}{m^2},
             \frac{\lambda^2}{\Delta^2},\frac{\Delta}{Q}\right).
 \label{eq:real_finite_ctp}
\end{equation}

The regulator-dependent and finite pieces are therefore
\begin{align}
 \mathcal B_{\mathrm V}^{\mathrm{SL}}
 &=-2(L-1)\log\frac{m}{\lambda}
   +[\mathcal B_{\mathrm V}^{\mathrm{SL}}]_{\rm fin},&
 [\mathcal B_{\mathrm V}^{\mathrm{SL}}]_{\rm fin}
 &=-\frac12L^2+\frac32L-2+\frac{\pi^2}{6},\label{eq:BvSL_finite}\\
 \mathcal B_{\mathrm V}^{\mathrm{TL}}
 &=-2(L-1)\log\frac{m}{\lambda}
   +[\mathcal B_{\mathrm V}^{\mathrm{TL}}]_{\rm fin},&
 [\mathcal B_{\mathrm V}^{\mathrm{TL}}]_{\rm fin}
 &=-\frac12L^2+\frac32L-2+\frac{2\pi^2}{3},\label{eq:BvTL_finite}\\
 \mathcal B_{\mathrm R}
 &=+2(L-1)\log\frac{m}{\lambda}
   +[\mathcal B_{\mathrm R}]_{\rm fin},&
 [\mathcal B_{\mathrm R}]_{\rm fin}
 &=2(L-1)\log\frac{2\Delta}{m}
   -\frac12L^2+L-\frac{\pi^2}{3}.
 \label{eq:Br_finite}
\end{align}
The $\log(m/\lambda)$ terms cancel exactly.  The finite real coefficient includes the physical resolution dependence through $\Delta$; it is not a regulator-dependent remainder.

\subsection{Two finite inclusive results}

For electron scattering from a nonradiating external Coulomb source, impose the soft-energy cut in the Breit frame and set $Q^2=-t$.  Adding Eqs.~\eqref{eq:BvSL_finite} and \eqref{eq:Br_finite} gives
\begin{equation}
 \begin{aligned}
 \frac{p_{\mathrm{Coul},\mathrm{V+R}}^{(1)}}{p_{\mathrm{Coul}}^{(0)}}
 &=\frac{\alpha}{\pi}\,\mathcal B_{\mathrm{Coul}}^{\mathrm{V+R}},\\
 \mathcal B_{\mathrm{Coul}}^{\mathrm{V+R}}
 &=2(L-1)\log\frac{2\Delta}{Q}
 +\frac32L-2-\frac{\pi^2}{6}
 \end{aligned}
 \qquad
 \left(L=\log\frac{Q^2}{m^2}\right).
 \label{eq:coulomb_finite_ctp}
\end{equation}
This is the correction multiplying the Born electron-line probability in the stated high-energy, soft-cut approximation.  Source radiation, vacuum polarization, and power-suppressed $F_2$ terms are separate contributions.

For $e^+e^-\to\mu^+\mu^-$, retain only the gauge-invariant final-state muon subset and impose $k^0<\Delta$ in the muon-pair center-of-mass frame.  Adding Eqs.~\eqref{eq:BvTL_finite} and \eqref{eq:Br_finite} yields
\begin{equation}
 \begin{aligned}
 \frac{p_{\mathrm{FSR},\mathrm{V+R}}^{(1)}}{p_{\mathrm{FSR}}^{(0)}}
 &=\frac{\alpha}{\pi}\,\mathcal B_{\mathrm{FSR}}^{\mathrm{V+R}},\\
 \mathcal B_{\mathrm{FSR}}^{\mathrm{V+R}}
 &=2(L-1)\log\frac{2\Delta}{\sqrt{s}}
 +\frac32L-2+\frac{\pi^2}{3}
 \end{aligned}
 \qquad
 \left(L=\log\frac{s}{m_\mu^2}\right).
 \label{eq:pair_finite_ctp}
\end{equation}
The result agrees with the high-energy limit of the exact final-state virtual-plus-soft factor~\cite{Czyz:2004}.  Initial-state radiation, initial--final interference, vacuum polarization, and hard resolved photons are not included in Eq.~\eqref{eq:pair_finite_ctp}.

Equations~\eqref{eq:coulomb_finite_ctp} and \eqref{eq:pair_finite_ctp} exhibit the finite content of the doubled contour.  The same-branch terms supply the renormalized virtual coefficient $\mathcal B_{\mathrm V}$; the cross-branch Wightman term supplies the unresolved real coefficient $\mathcal B_{\mathrm R}$; and the auxiliary photon mass cancels only after both CTP sectors are added.  Every displayed probability correction carries the single overall factor $\alpha/\pi$ defined in Eq.~\eqref{eq:B_coefficient_definition}.  The timelike and spacelike answers differ by the finite analytic-continuation term $\pi^2/2$, while their infrared kernels are identical.

\end{document}